\documentclass[12pt]{article}

\usepackage{epsfig}

\usepackage{amssymb}
\usepackage{amsfonts}

\usepackage{color}
 
\def\be{\begin{equation}}
\def\ee{\end{equation}}
\def\ba{\begin{array}{c}}
\def\ea{\end{array}}

\newcommand{\kt}{\rangle}
\newcommand{\br}{\langle}

\newtheorem{thm}{Theorem}
\newtheorem{cor}[thm]{Corollary}
\newtheorem{lemma}[thm]{Lemma}

\newenvironment{proof}{\noindent
 {\bf Proof.}}{\hfill$\square$\vspace{3mm}\endtrivlist}

\begin{document}

\begin{center}

{\Large

Broken Hermiticity phase transition in Bose-Hubbard model

}

\vspace{0.8cm}

  {\bf Miloslav Znojil}

\vspace{0.2cm}

\vspace{1mm} Nuclear Physics Institute of the CAS, Hlavn\'{\i} 130,
250 68 \v{R}e\v{z}, Czech Republic

{e-mail: znojil@ujf.cas.cz}

\end{center}

\section*{Abstract}

For the two-mode and $(N-1)-$bosonic Bose-Hubbard quantum system a
less usual phase transition controlled by the parameter
$\varepsilon$ representing the on-site energy difference is studied.
In the literature the parameter is considered either real
($\varepsilon>0$) or purely imaginary (with, say, $\gamma={\rm Im}\
\varepsilon>0$), so the phase transition is analyzed here at the
interface $\varepsilon=\gamma=0$. The evolution in the
$\gamma-$controlled phase is required unitary so that the main task
for the theory is found in the (quasi-)Hermitization of the
Hamiltonian, achieved by a suitable amendment of the inner product
in Hilbert space, $\br \cdot|\cdot\kt \to \br
\cdot|\Theta|\cdot\kt$. In the most relevant domain of small
$\gamma$ the linearized Hilbert-space metric $\Theta(\gamma)$
(constrained by the requirement $\lim_{\gamma \to
0}\Theta(\gamma)=I$ of the smoothness of the change of the Hilbert
space at the phase transition) is constructed in closed form. Beyond
the phase-transition instant, several forms of the systematic
non-numerical recurrent construction of the exact metrics
$\Theta(\gamma)$ are also shown user-friendly and feasible, at the
not too large matric dimensions $N$ at least.

\subsection*{Keywords}

quantum phase transitions; Bose-Hubbard model; non-Hermitian ${\cal
PT}-$symmetric phase; unitary evolution; {\it ad hoc\,}
Hilbert-space metrics; non-numerical construction methods;

\newpage

\section{Introduction}

Bosonic version of the Hubbard model is called Bose-Hubbard model
\cite{bose}. It describes the zero-spin particles on a lattice at
zero temperature in a way which is well adapted, in the context of
solid state physics, to the study of the phase transition between
its superfluid and insulator phases induced by the variation of the
density \cite{boseb}. In the conventional many-body setting the
Bose-Hubbard (BH) Hamiltonian is, typically, able to deal with the
Bose-Einstein condensation \cite{BEC}. The amendments of the model
can also offer a theoretical background to various other forms of
the quantum phase transitions, say, in optical lattices \cite{OL}.

The well known user-friendly mathematical tractability of the model
\cite{Zhang} can be perceived as originating from its Lie-algebraic
background. Thus, one can write the two-mode version of the
conventional self-adjoint BH Hamiltonian in terms of the two
angular-momentum generators $L_{x,z}$ of Lie algebra $su(2)$ using
just three real parameters $\varepsilon, v$ and $c$ \cite{Uwe},
 \be
 \mathfrak{h}_{(BH)}(\varepsilon,v,c)
 =2\varepsilon\,L_z+2v\,L_x+ 2c\,L^2_z
 =
  \mathfrak{h}^\dagger_{(BH)}(\varepsilon,v,c)
 \,.
 \label{hermbh}
 \ee
Obviously, the most efficient treatment of the changes caused by the
variations of the physical quantity $2c$ representing the strength
of the interbosonic interactions will be provided by perturbation
theory. Still, even if we restrict attention to the zero-order
approximation, we are left with the variability of the two
independent parameters, viz., of the quantity $2v$ which measures
the intensity of the single-particle tunneling, and of the value
$2\varepsilon$ which characterizes the bosonic on-site energy
difference. One of these quantitites may be fixed via a suitable
choice of the units. Thus, once we set, say, $v=1$, we only have to
study the one-parametric problem.

In such a setting the authors of Refs.~\cite{Uwe,Uweb} imagined that
it is far from obvious that the parameter in question must be real.
They gave several tenable arguments supporting the study of the
possible inclusion of non-Hermiticities. In particular, the authors
of Ref.~\cite{Uwe} proposed the replacement of the bosonic on-site
energy by a purely imaginary quantity,
 \be
 2\varepsilon \to 2{\rm i} \gamma\,.
 \label{modifi}
 \ee
Naturally, the change opened a Pandora's box of interpretational
challenges. The main one was that the new, complexified Bose-Hubbard
(CBH) Hamiltonian ceased to be a self-adjoint operator,
 \be
 H_{(CBH)}(\gamma,v,c)=
  -2\,{\rm i}\gamma \,L_z +2v\,L_x+ 2c\,L^2_z\ \neq \
  H_{(CBH)}^\dagger(\gamma,v,c)\,.
  \label{Uwemo}
 \ee
In ~\cite{Uwe} the problem has been settled by an
open-quantum-system upgrade of the underlying physics. In essence,
an {\it ad hoc\,} external field has been assumed to mimic the
influence of the environment causing the parameter-controlled gains
and/or losses of the bosons.

The authors of the idea felt inspired by the recent growth of
interest in the Hamiltonians which are non-self-adjoint but ${\cal
PT}-$symmetric (cf., e.g., reviews \cite{Carl,ali,book}). On this
background it was possible to conclude that in the CBH model one can
clearly distinguish between its ``stable'' and ``unstable''dynamical
regimes, separated by a new form of phase transition. In the former
case, indeed, all of the eigen-energies remain real because the
${\cal PT}-$symmetry of the system is observed not only by the
Hamiltonian but also by its eigen-states. In the ``unstable'' case,
on the contrary, the ${\cal PT}-$symmetry becomes spontaneously
broken. This means that some of the energies complexify while the
related eigen-states cease to be ${\cal PT}-$symmetric.

The CBH-related research found one of its central topics in the
study of the ``instants'' of the breakdown of ${\cal PT}-$symmetry.
The existence of such ``exceptional points'' (EP) in the analytic
quantum Hamiltonians is well know to mathematicians \cite{Kato}.
Still, the occurrence and the role of EPs in various physical
systems has only been clarified rather recently \cite{Heiss}. In
particular, the careful localization of the EP singularities
$\gamma^{(EP)}_{(CBH)}(v,c)$ helped the authors of Ref.~\cite{Uwe}
to clarify further the connection between the CBH Hamiltonians
(\ref{Uwemo}) and the Bose-Einstein condensation.

In \cite{45} we pointed out that besides the CBH class of the
phenomenological models it is also possible to construct and use
their various complex-symmetric $N$ by $N$ matrix generalizations
for the same phenomenological purposes and in the same EP-related
context. We concluded that one of the most characteristic physical
features of {\em any\,} Hamiltonian of the EP-supporting type is
that in the apparently most interesting EP-controlled dynamical
regime in which the system can perform a quantum phase transition
the operator itself is {\em strongly\,} non-Hermitian \cite{Denis}.
This leads to a rather paradoxical situation in which the study of
the {\em weakly\,} non-Hermitian regime (which is closer to the
conventional Hermitian regime) is almost completely neglected in the
literature. Now we intend to fill the gap. On an entirely abstract
conceptual level such a project is promising because the popular
restriction of attention to the strongly non-Hermitian EP-related
quantum phase transitions is unnecessarily restrictive
\cite{Borisov}. Here, we shall accept a different, less restrictive
philosophy.

Our present project is inspired by our methodological study
\cite{interface} in which a new model-building strategy has been
outlined. In essence, we described there the new type of an
EP-unrelated quantum phase transition using just a
weak-non-Hermiticity mathematical background. Our return to this
subject was recently re-encouraged when we noticed that the
phenomenon of the EP-unrelated quantum phase transition was also
revealed and predicted in several non-Hermitian
multidimensional-oscillator examples \cite{Fern} as well as in the
realistic-physics context of the dissipative photonic systems
\cite{Joglekar} and/or, on experimental level of classical-physics
simulations, of synthetic circuits \cite{Joglekarb}.

In \cite{interface} our specific interface-passage considerations
were illustrated by a schematic two-by-two matrix Hamiltonian. We in
fact did not pay too much attention to the details of the evolution
{\em after\,} the passage. We were aware that a realistic
illustration would be highly desirable and that such an illustration
might have been provided by the CBH Hamiltonians (\ref{Uwemo}).
Still, we felt that the task might be prohibitively complicated,
especially because after the passage the consistent description of
the stable evolution of the system in question would {\em
necessarily\,} require the {\em explicit\,} construction of the so
called operator of charge \cite{Carl} or, in a more general
quasi-Hermitian setting of Refs.~\cite{Geyer,timedep}, of the so
called Hilbert-space metric $\Theta$.

Only recently, having reread section \# 3 of paper \cite{Uwe} we
imagined that a compromising solution might have been sought, and
the purpose could have been served, by the simplified, unperturbed
CBH model with $c=0$. The idea proved productive and it led to the
results presented in what follows. Their presentation will be
preceded by section \ref{kap2} in which the reader finds a compact
outline of the current stage of development of the concept of
quantum phase transition, with special emphasis upon its
non-Hermitian descriptions. In subsequent section \ref{ohub} we
shall describe the necessary technical aspects of the CBH model in
its separate finite-dimensional $N$ by $N$ matrix representations. In
particular, we shall point out that the consistent presentation of
the model necessitates, via the construction of $\Theta_{(CBH)}$,
the {\em explicit\,}  specification of the ``standard''
physical Hilbert space ${\cal H}^{(S)}_{(CBH)}$ in which one only
can clarify the notion of the observability \cite{ali}.

It is worth re-emphasizing that we will treat the CBH Hamiltonian
(\ref{Uwemo}) (with $c=0$) as the quasi-Hermitian operator
\cite{Geyer,Dieudonne}, i.e., as the generator of the evolution
which is {\em unitary} in ${\cal H}^{(S)}$. In contrast to the
open-system theories~\cite{Nimrod}, our present version of the CBH
model will be built differently, as a closed quantum system without
any implicit or explicit reference to an interaction with the
environment. The feasibility of such a project will be facilitated
by the solvability of the model guaranteeing the reality (i.e., the
potential observability) of the spectrum in a sufficiently large
interval of $\gamma \in (0,\gamma_{\max})$.

As we already indicated, our main (and, in practice, almost always
most difficult) technical task will be the construction of the
after-the-transition Hamiltonian-dependent Hilbert space ${\cal
H}^{(S)}={\cal H}^{(S)}_{(CBH)}(\gamma)$ or, more precisely, of its
acceptable physical inner product. This will be done in section
\ref{uukol} (for the first nontrivial choice of $N=3$), in section
\ref{uxkol} (for $N=4$) and in section \ref{vvkol} [where we will
discuss the extrapolation of our knowledge to all $N$, with tests
performed at $N=5$ (in subsection \ref{5kol}) and $N=6$ (in
\ref{mol})]. Finally, our message will be summarized in section
\ref{summat}.

\section{Hermitian - quasi-Hermitian phase transition\label{kap2}}

\subsection{Analytic Hamiltonians and phase transitions at exceptional
points}

During the study of the phenomena called quantum phase transitions
one might often hesitate whether certain abrupt changes of
properties of a given system should still be given the name of phase
transition. One not always finds a guidance in parallels between the
classical and quantum physics \cite{Sachdev}. In one direction, for
a given quantum model it is not always easy to deduce the classical
$\hbar \to 0$ limit. The situation is even worse in the opposite
direction in which the correspondence principle leads to
quantization recipes which may be ambiguous \cite{Hall}.

Examples of the incompleteness of the parallels abound, especially
after the physics community accepted the idea that it may be useful
to study stable quantum systems in their non-Hermitian (usually
called ${\cal PT}-$symmetric {\it alias\,} pseudo-Hermitian)
representations (cf., e.g., the respective comprehensive reviews
\cite{Carl} and \cite{ali}). In such a framework several new
theoretical ideas emerged during the last 25 years (cf., e.g., the
introductory chapter in Ref.~\cite{book}). In 1992, for example, the
well known, exactly solvable $N-$fermion Lipkin-Meshkov-Glick model
of Ref.~\cite{LMG} has been generalized, by Scholtz et al
\cite{Geyer}, in a way which sampled, in the context of
many-particle quantum physics, several new forms of phase
transitions. Between 1997 and 1998 Bender with coauthors
\cite{BM,BB} discovered, in a different context of quantum field
theory, an equally interesting class of innovative quantum phase
transitions which they called the spontaneous breakdown of ${\cal
PT}$ symmetry.

These discoveries were followed by the identification of several
quantum phase transition phenomena reflecting the presence of the
Kato's exceptional point (EP, \cite{Kato}) in the Hamiltonian. In
the most frequently encountered phase transition of this type the
energies become complex so that the physical interpretation  of the
original, unitarily evolving and stable quantum system was lost. For
the similar situations it is characteristic that one must introduce
some new degrees of freedom so that the initial Hamiltonian
$\Lambda_0^{(before)}$ as well as at least some of the other
observables $\Lambda_1^{(before)}, \Lambda_2^{(before)}, \ldots$
must be replaced, after the passage of the system in question
through its EP singularity, by some entirely different operators
$\Lambda_j^{(after)}$ with $j=0, 1, \ldots$.

In general, the description of the latter (also known as ```first
kind'') quantum phase transition requires also the change of the
underlying physical Hilbert space, ${\cal H}^{(before)} \to {\cal
H}^{(after)}$. Still, for some rather special quantum systems there
also exist exceptions. In these cases one can admit the survival of
kinematics (with ${\cal H}^{(before)} \equiv {\cal H}^{(after)}$) as
well as of the dynamics (controlled by {\em the same}, EP-possessing
and unchanged non-Hermitian Hamiltonian). This scenario (called
``quantum phase transition of the second kind'', cf. \cite{Borisov})
is characterized by the mere {\em partial\,} loss of the
observability involving just a {\em subset\,} of all of the relevant
$\Lambda_j$s with $j\neq 0$.

In the latter scenario the evolution may be required to remain
unitary. This is rendered possible by the fact that after the phase
transition the bound-state energies are still real and observable.
Nevertheless, as long as the change always involves at least some of
the $j>0$ observables $\Lambda_j^{(before)}$, one encounters a full
freedom in the choice of their descendants $\Lambda_j^{(after)}$. As
a consequence, the change of operators $\Lambda_j^{(before)} \to
\Lambda_j^{(after)}$ (which are all, under our overall
unitary-evolution hypothesis, {\em necessarily\,} quasi-Hermitian
\cite{Geyer}) may imply, in general, a parallel consistent change of
the metric operator, $\Theta^{(before)} \neq \Theta^{(after)}$. An
analogous scenario will be also used and built in our present paper.

\subsection{The instant of onset of non-Hermiticity\label{osum}}

In the toy-model of Ref.~\cite{interface} the conventional
Hermiticity of observables was lost and replaced,  at an {\it ad
hoc\,} phase-transition interface, by the so-called
quasi-Hermiticity. Several features of the passage of the system
through such a boundary were discussed, with emphasis upon the
methodological aspects of the problem. The passage from the Hermitian to
quasi-Hermitian dynamical regime was illustrated by the most
elementary two-by-two-matrix toy model. Our Hamiltonian $H$ was
time-dependent and non-Hermitian but ${\cal PT}-$symmetric, with
real spectrum. Its elementary nature helped us to clarify the basic
features of the mathematically correct treatment of the dynamics of
the system. The question of a more realistic physical applicability
of the formalism remained open.

Let us now return to Eq.~(\ref{Uwemo}) representing one of the most
interesting non-Hermitian but still deeply realistic Hamiltonians.
In subsequent sections we shall review some of the basic properties
of the model, emphasizing the difference between its two possible
physical probabilistic interpretations. We shall explain that in a
way outlined in Refs.~\cite{Nimrod} and \cite{ali} this difference
reflects the freedom of the choice between the non-unitarity and
unitarity of the evolution or between the theoretical framework of
the open and closed quantum system, respectively.

We shall restrict our attention to the unitary case. It has an
advantage that our knowledge of the dynamics is complete, not
involving any hypothetical environment. Due to the ``hidden'' form
of the Hermiticity of the observables the dynamical information
about the evolution may be carried not only by the Hamiltonian but
also by the above-mentioned operator $\Theta$. Indeed, the latter
Hilbert-space-metric operator carries such an information because it
determines the correct physical inner product in the standard
Hilbert space ${\cal H}^{(S)}$ \cite{SIGMA}.

Needless to add that in many models a guarantee of the compatibility
between the information carried by $H$ and $\Theta$ can be
nontrivial \cite{arabky}. In fact, the necessity of this guarantee
has been perceived, in the past, as one of the key obstructions of
the applicability of the pseudo-Hermitian {\it alias\,} ${\cal
PT}-$symmetric constructions in realistic situations.

\subsection{Non-Hermitian phase and the unitarity of its evolution }

The authors of paper \cite{Uwe} circumvented the search for the
Hilbert space  ${\cal H}^{(S)}$ via the open-system physical
treatment of their manifestly non-Hermitian quantum CBH
Hamiltonians. They only studied the localization of the spontaneous
breakdown of ${\cal PT}-$symmetry in the strongly non-Hermitian
dynamical regime. In this case, the parameter $\gamma$ measuring the
strength of the non-Hermiticity was assumed large, close to its
maximal, transition-responsible EP value $\gamma^{(EP)}$.

Incidentally, we should add that even in the latter, truly extreme
dynamical regime it should still be possible to follow the
unitary-evolution philosophy and constructions, in principle at
least. In practice one can of course expect that these constructions
will be perceivably easier in the weakly non-Hermitian regime.

In the latter regime the main phenomenological advantages of the
closed-system approach are twofold. Firstly, the unitary picture of
the evolution generated by matrices $H^{(N)}_{(CBH)}(\gamma)$ is
{\em complete}. There is no need of referring to an unspecified
environment \cite{Jones}. In ${\cal H}^{(S)}$ the underlying quantum
theory becomes fully compatible with the conventional textbooks
\cite{Messiah}. Secondly, the conservation of the unitarity during
the non-EP phase transitions may open the way towards a matching of
two alternative Hilbert-space representations of quantum world in a
unified picture.

The replacement (\ref{modifi}) may acquire, in this spirit, a
smooth-transition meaning at an interface where
$\varepsilon=\gamma=0$. For the sake of definiteness let us agree
that we shall follow the very slow adiabatic change of the initially
Hermitian BH model (\ref{hermbh}) in the limit of vanishing small
$\varepsilon \to 0$. After the system touches the interface and
after it performs the Hermiticity - non-Hermiticity phase transition
(\ref{modifi}), the subsequent evolution will be controlled by the
complexified model (\ref{Uwemo}), with a small $\gamma$, i.e., with
the small {\em imaginary\,} on-site energy difference.

In the new regime one must discuss several mathematical components
of the model requiring, e.g., the reality of the energies and,
secondly, a phenomenologically sufficiently well motivated choice of
the correct physical Hilbert space. The latter process involves not
only an explicit and correct constructive assignment of the metric
$\Theta=\Theta(H)$ to the preselected Hamiltonian $H$ but also an
appropriate suppression of the ambiguity of such an assignment (cf.,
e.g., \cite{lotor} for an eligible ``minimality'' criterion and
technique of such a suppression).

\section{Hermitian vs. quasi-Hermitian Bose-Hubbard model
\label{ohub}}

\subsection{The matching at $\varepsilon=\gamma=0$}

After the introductory considerations let us now turn
attention back to the exactly solvable interaction-free (i.e.,
$c=0$) Bose-Hubbard model in its versions~(\ref{hermbh})
and~(\ref{Uwemo}). No special remarks have to be added to the former
case which is an entirely conventional Hermitian model. In our
present notation it will be assigned, at any value of its real
variable parameter $\varepsilon \geq 0$, the trivial metric
$\Theta_{Hermitian}=I$.

Once we intend to speak about the Hermiticity/quasi-Hermiticity
quantum phase transition in the limit $\varepsilon \to 0$ (i.e.,
say, at time $t=0$), we only have to describe the post-transition
quantum system using the elementary complex upgrade (\ref{modifi})
of the parameter at any $t>0$. The unitarity of the $t>0$ evolution
may be then guaranteed by the proper interpretation of
Schr\"{o}dinger equation in an adiabatic approximation \cite{NIP},
i.e., for the sufficiently small times and $\gamma$s at least
\cite{Doppler}. Our upgraded quasi-Hermitian Bose-Hubbard
Hamiltonian (\ref{Uwemo}) will be, subsequently, assigned a
nontrivial metric $\Theta=\Theta(H)$.

The resulting pair of operators $H_{(CBH)}(\gamma)$ and
$\Theta_{(CBH)}(\gamma)$ describing the dynamics after $t=0$ must
necessarily satisfy the Dieudonn\'{e}'s constraint
 \be
 H^\dagger \Theta=\Theta H\,.
  \label{requi}
 \ee
In what follows we intend to satisfy such a requirement
constructively. Our construction will be facilitated by the
knowledge of some basic properties of operators $H_{(CBH)}(\gamma)$
as provided by the authors of paper \cite{Uwe}. As long as they did
not pay attention to the construction of $\Theta_{(CBH)}(\gamma)$,
we will have to address the following questions:

\begin{itemize}

\item
we will have to solve Dieudonn\'{e}'s Eq.~(\ref{requi}) interpreted
as an implicit definition of metric;

\item
as long as the latter definition is ambiguous \cite{SIGMAdva}, we
will restrict the class of solutions to the subclass of the
candidates for metric which are positive definite;

\item
we will have to reduce the resulting set of eligible metrics
$\Theta=\Theta_j(\gamma)$ to such a sub-family
$\Theta_j^{(0)}(\gamma)$ for which the phase transition at $t=0$
would be smooth, i.e., for which the metric candidate trivializes in
the limit $\gamma\to 0^+$, i.e.,
 \be
 \lim_{\gamma \to 0^+}\Theta_j^{(0)}(\gamma)=I\,;
 \ee

\item
subsequently, we will be allowed to fix all of the remaining free
parameters arbitrarily.

\end{itemize}

\subsection{Quasi-Hermitian dynamical regime with
$\gamma>0$\label{solohub}}

The unperturbed models with $c=0$ remain exactly solvable in both
the Hermitian and non-Hermitian cases (cf. section \# 3 in
\cite{Uwe}). Their mutual relationship (\ref{modifi}) is usually
perceived as purely formal. In Ref.~\cite{Uwe}, for example, the
authors claim that the study of complexified model (\ref{Uwemo}) is
most interesting in the strongly non-Hermitian domain, near the
above-mentioned Kato's exceptional points $\gamma^{(EP)}$, and far
from the onset-of-non-Hermiticity limit $\gamma \to 0$. This makes a
false impression that the model is only interesting very far from
its possible phase-transition-like transmutation into its
self-adjoint partner (\ref{hermbh}), say, at a hypothetical
interface with $\varepsilon\approx 0 \approx \gamma$. We believe
that the opposite is true. A truly exciting physics might be
expected to emerge at small $\varepsilon$ and $\gamma$. First of
all, the passage of a quantum system in question through the
$\gamma=\varepsilon=0$ interface would be a rather unusual and
specific quantum phase transition. Secondly, the study of passages
through an interface of such a type might throw new light on the
range of validity of the recent discoveries of the failure of
adiabatic hypothesis in the strongly non-Hermitian regime
\cite{Doppler}. Thirdly, the study of the interface in a weakly
non-Hermitian representation might prove technically feasible.

\subsubsection{Representation by matrices}

In the light of the representation theory of $su(2)$ the toy-model
operator (\ref{Uwemo}) may be decomposed into an infinite family of
its finite-dimensional $N$ by $N$ matrix representations
$H^{(N)}_{(CBH)}(\gamma)$. They may be sampled by the $N=2$ matrix
 \be
 H^{(2)}_{(CBH)}(\gamma)=
 \left[ \begin {array}{cc} -i{\it \gamma}&1
 \\\noalign{\medskip}1&i{\it
 \gamma}
 \end {array} \right]\,
 \label{dopp2}
 \ee
(with the bound-state spectrum $E_\pm = \pm \sqrt{1-\gamma^2}$ and
EPs $\gamma^{(EP)}_\pm=\pm 1$) or, at $N=3$, by
 \be
H^{(3)}_{(CBH)}(\gamma)=\left[ \begin {array}{ccc} -2\,i\gamma&
\sqrt{2}&0\\\noalign{\medskip}\sqrt{2}&0&
\sqrt{2}\\\noalign{\medskip}0&\sqrt{2}&2\,i\gamma\end {array}
\right]\,
  \label{3wg}
 \ee
(with $E_0=0$, $E_\pm = \pm 2\,\sqrt{1-\gamma^2}$ and
$\gamma^{(EP)}_\pm =\pm 1$), etc. The key formal advantage of the
complex symmetric structure of these matrices has been found, in
Refs.~\cite{admis}, in a technically friendly nature of the
incorporation of perturbation corrections. In this framework the
authors of \cite{Uwe} studied the dynamical regime in which  $\gamma
\approx \gamma^{(EP)}$ and in which the bosons sit in a double-well
potential endowed with the respective sink- and source-simulated
additional couplings to an external continuum. In this language they
were really able to clarify certain specific features of the
Bose-Einstein condensation treated as an EP-related phase
transition.

\subsubsection{Physical inner product}

It is easy to show that at $c=0$ and all $N$, the real line of the
CBH parameter $\gamma$ splits into an open interval ${\cal
D}=(-1,1)$ (in which the spectrum is real and non-degenerate), the
complement $(-\infty,-1)\bigcup (1,\infty)$ of the closure of ${\cal
D}$ (in which the spectrum is not real), and the boundary $\partial
{\cal D}=\{-1,1\}$  formed by the Kato's exceptional points
$\gamma^{(EP)}_\pm = \pm 1$.
In our present paper we will exclusively pay attention to the
interior of ${\cal D}$, treating Hamiltonian (\ref{Uwemo}) as a
standard self-adjoint operator of a standard quantum observable.

In the spirit of Stone's theorem \cite{Stone} the latter operator
plays just the role of the generator of unitary evolution. After the
phase transition such a role will be rendered possible by the
replacement of the conventional ``friendly but false'' Hilbert space
${\cal H}^{(F)}$ (endowed with inner product $\br \cdot|\cdot \kt$)
by its ``standard'' physical amendment ${\cal H}^{(S)}$ in which the
Banach-space topology remains unchanged and only the correct inner
product is  different,
 $$
 \br \cdot|\cdot \kt\ \to \  \br \cdot|\Theta|\cdot \kt\,.
 $$
Once we are given {\em any\,} diagonalizable
Hamiltonian with real spectrum, the construction of the necessary
physical Hilbert space can be perceived, in our present
finite-dimensional cases at least, as equivalent to the construction
of the Hermitizing operator $\Theta=\Theta(H)$ via Dieudonn\'{e}'s
equation (\ref{requi}).
In such a unitary-evolution approach also every CBH matrix (\ref{dopp2}),
(\ref{3wg}) (etc) must be assigned its sophisticated inner product
rendering
this matrix Hermitian. Without an additional information about
dynamics, we may simply consider {\em any\,} set of the positive definite
solutions of the Dieudonn\'{e}'s Eq.~(\ref{requi}) and declare
any one of them ``physical''.

\subsubsection{Metrics at $N=2$}

The insertion of Hamiltonian (\ref{dopp2}) converts
Eq.~(\ref{requi}) into the definition of all of the eligible CBH
Hilbert-space metrics at $N=2$ \cite{45},
 \be
 \Theta^{(2)}(\beta)= I^{(2)} +\left[ \begin {array}{cc}
 0&\beta+i\,\gamma
 \\
 \noalign{\medskip}\beta-i\,\gamma&0
 \end {array} \right]\,,\ \ \ \ \ -\sqrt{1-\gamma^2}<\beta<
 \sqrt{1-\gamma^2}\,.
 \label{medopp2}
 \ee
The new free real parameter $\beta$ numbers the complete set of the
different physical inner products, i.e., the different classes of
the eligible observables $\Lambda=\Lambda^{(2)}_j(\beta)$ which must
be compatible with the same metric, i.e., which must be self-adjoint
in the same Hilbert space ${\cal H}^{(S)}$,
 \be
 \Lambda^\dagger_j(\beta) \Theta^{(2)}(\beta)=\Theta^{(2)}({\beta})
 \Lambda_j(\beta)\,, \ \ \ \ j = 1, 2, \ldots\,, J\,.
 \,.
  \label{requiem2}
 \ee
Naturally, we must satisfy conditions (\ref{requiem2}) even if we
decide to pre-select the physics-dictated observables
$\Lambda^{(2)}_j$ in advance. Nevertheless, the metric $\Theta$
compatible with all of them need not then exist at all
\cite{arabky}.

In the conventional, Hermitian quantum mechanics the ambiguity of
the assignment $H \to \Theta(H)$ is ignored. The ``formally
optimal'' choice of trivial $\Theta=I$ is practically never put
under questionmark. The situation is different for the non-Hermitian
Hamiltonians with real energy spectra because it is much less
obvious which one of the available Hermitizations based on the
necessarily nontrivial inner product is ``optimal''. One could, for
example, follow our recent recommendation \cite{lotor} and require
that the difference $\Theta - I $ should be kept, in some sense,
minimal.

This may be also required at $N=2$. The most natural measure of the
difference between $\Theta$ and the unit operator $I$ may be then
based on our knowledge of all of the eigenvalues of the general CBH
metric,
 \be
 \theta_{\pm}^{(2)}(\beta)= 1\pm
 \sqrt{\beta^2+\gamma^2}\,.
 \label{alleig}
 \ee
The difference between these eigenvalues specifies, in the language
of geometry, the extent of anisotropy in the vector space
$\mathbb{C}^{2}(\Theta)$. The ``minimal anisotropy choice'' would be
unique, achieved at constant $\beta=0$. This value is also
sufficient and acceptable in the whole interval of $\gamma \in
(0,1)$ [or rather of $\gamma \in (-1,1)$] in which the energies
remain real and, hence, potentially observable.

\section{Unitarity of evolution after the phase
transition ($N=3$)\label{uukol}}


\subsection{Matrices of the metric at small $\gamma$}

After we abbreviate $G=\sqrt{2}\,\gamma$ we may insert
Eq.~(\ref{3wg}), together with a general ansatz for $\Theta$, in
Eq.~(\ref{requi}). The two-parametric
solution is immediate \cite{45},
 \be
 \Theta^{(3)}(\beta,\delta)= I^{(3)} +\left[ \begin {array}{ccc}
  0&\beta+i\,G&
 \delta+i\,G\,\beta\\
 \noalign{\medskip}\beta-i\,G&\delta+{G}^{2}&\beta+i\,G\\
 \noalign{\medskip}\delta-i\,G\,\beta&\beta-i\,G&0\end {array}
 \right]\,.
  \label{me3wg}
 \ee
The guarantee of the positive definiteness does have a purely
algebraic form, in principle at least \cite{positivity}.
In practice the localization of the boundaries of
the domain of positivity of the metric-operator roots remains
prohibitively complicated because the underlying secular polynomial
contains as many as 21 separate terms.

For our present purposes it will be sufficient to construct the
metric at small $\gamma$. We will only have to require that the
onset of the non-Hermiticity should be smooth, i.e., that the metric
should not differ too much from the unit matrix. The inspection of
Eq.~(\ref{me3wg}) reveals that the free parameters $\beta$ and
$\delta$ should be then small.

Once we start
from $\beta=\delta=0$ we get the exact eigenvalues of $\Theta$ in
closed form. Besides the constant $\theta_0=1$, the
other two roots
 \be
  \theta_\pm =
 1\pm \frac{1}{2}\,\sqrt
 {8\,{G}^{2}+{G}^{4}}+\frac{1}{2}\,{G}^{2}
 = 1 \pm \gamma \sqrt{4+\gamma^2})+\gamma^2\,,
 \ \ \ \  -1/\sqrt{2}<\gamma < 1/\sqrt{2}
 \,
 \label{someig}
 \ee
of the underlying secular polynomial are $\gamma-$dependent and
positive. In the weakly non-Hermitian dynamical
regime, all of the eigenvalues of the most general metrics
(\ref{me3wg}) remain perturbatively close to one.

The simplest metric with trivial $\beta=\delta=0$ ceases to
be invertible at the
reasonably large value of
$\gamma_{critical}(\beta,\delta)=1/\sqrt{2} \approx 0.707$.

\subsection{Corrections}

In the search for corrections let us set $\theta=1+s$ and get the
polynomial secular equation
 $$
 0={s}^{3}-{\it {\delta}}\,{s}^{2}+ \left( -2\,{{\it
 {\beta}}}^{2}-4\,{{\it
\gamma}}^{2} -{{\it {\delta}}}^{2} \right) s
 - 2\,{\it {\delta}}\,{{\it
{\beta}}}^{2}+4\,{\it {\delta}}\,{{ \it \gamma}}^{2}+{{\it
{\delta}}}^{3}\,
 $$
in which we omitted all of the fourth- and higher-order terms. This
polynomial is quadratic in some parameters making the approximate
solutions expressible in the elementary implicit-function forms
$\gamma=\gamma(s,\beta,\ldots)$ or $\beta=\beta(s,\gamma,\ldots)$.
Although the resulting formulae are slightly tedious, their
graphical and/or Taylor-series analysis remains straightforward. At
$\delta=0$ one obtains, in particular, the same elementary
parameter-dependence as above,
 $$
 s_0=0\,,\ \ \ \ \ \
 s_\pm = \pm \sqrt{2\,\beta^2+4\,\gamma^2}\,.
 $$
The message delivered by such an
analysis is encouraging because at the higher
matrix dimensions we may expect that

\begin{itemize}

\item
the candidates for the metric may be written in the form of the unit
matrix plus corrections,

\item
the correction matrix may be sought in the form which is complex symmetric
with respect to the second diagonal.

\end{itemize}


\subsection{The metric covering the whole range of admissible
$\gamma$}

In contrast to the preceding $N=2$ model our $\beta=\delta=0$
formula (\ref{someig}) fails to hold up to the EP-supremum
of $\gamma \to \gamma^{(EP)}=1$. The failure may be attributed to
the trivial choice of parameter $\delta$. This may be tested at
$\gamma=\gamma_{critical}(0,0)=1/\sqrt{2}$. At this value of $\gamma$,
our
two-parametric family of metric candidates
 $$
 \Theta_{\gamma=1/\sqrt{2}}(\beta,\delta)=
 \left[ \begin {array}{ccc} 1&{\it \beta}+i&{\it \delta}+i{\it
 \beta}
 \\ \noalign{\medskip}{\it \beta}-i&{\it \delta}+2&{\it \beta}+i
 \\ \noalign{\medskip}{\it \delta}-i{\it \beta}&{\it \beta}-i&1\end
 {array}
 \right]\,
  $$
may be assigned the secular polynomial
 $$
 P_{\gamma=1/\sqrt{2}}(s)={s}^{3}+ \left( -4-{\it \delta} \right)
 {s}^{2}+
 \left( 3-{{\it \delta}}^{2}-3 \,{{\it \beta}}^{2}+2\,{\it \delta}
 \right)
 s-{\it \delta}\,{{\it \beta}}^{2}+{{\it \delta}}^{3}+{\it
 \delta}+2\,{{\it
 \delta}}^{2}\,
 $$
for which all of the roots may happen to be positive at
a nontrivial $\delta$.

After the omission of the higher, subdominant powers of parameters
the localization of the amended
$\delta$s degenerates to the simplified secular equation
 $$
 0=
 {s}^{3}+ \left( -4-{\it \delta} \right) {s}^{2}+ \left( 3+2\,{\it
 \delta}
 \right) s+{\it \delta}
 $$
which, incidentally, does not depend on $\beta$ at all.
Numerically we localized its
three roots
which remained {\em positive\,}
inside a finite interval of negative
$\delta \gtrapprox -0.382478$.
On this ground one can
extend the admissible interval of $\gamma$ at the expense of using
a tentative and, presumably, negative $\gamma-$dependent function
$\delta=\delta(\gamma)$.

In the most elementary implementation of such an idea let us set
$\beta=0$ and $\delta=-\gamma$. This leads to a new
$\gamma-$dependent metric candidate
 $$
 \Theta^{(3)}(0,-\gamma)
 =\left[ \begin {array}{ccc} 1&i\sqrt {2}{\it \gamma}&-{\it \gamma}
 \\
 \noalign{\medskip}-i\sqrt {2}{\it \gamma}&
 1-{\it \gamma}+2\,{{\it \gamma}}^{2}&
 i\sqrt {2}{\it \gamma}\\\noalign{\medskip}-{\it \gamma}&-i\sqrt
 {2}{\it \gamma}&1
 \end {array} \right]
 $$
with the three (exact) eigenvalues
 $$
 \theta_1=1-\gamma\,,
 \ \ \
 \theta_{2,3}=
 1\pm \sqrt {5\,{{\it \gamma}}^{2}+{{\it \gamma}}^{4}-2\,{{\it
 \gamma}}
^{3}}+{{\it \gamma}}^{2}\,.
 $$
These eigenvalues are, obviously, positive in the whole (open)
interval of  physically relevant CBH parameters $\gamma \in
(0,\gamma^{(EP)})$ with $\gamma^{(EP)}=1$. The related metric
remains acceptable, therefore, up to the maximally non-Hermitian
domain of $\gamma \lessapprox \gamma^{(EP)}$.

The amended construction introduces a more pronounced anisotropy in
the ``standard'', metric-dependent physical Hilbert space ${\cal
H}^{(S)}$ at the smallest $\gamma$s. Thus, the amendment is less
suitable for the study of the phase-transition onset of the
non-Hermiticity at $\gamma \approx 0$.

\section{Constructions of metric at $N=4$\label{uxkol}}

The wisdom gained via the $N=2$ and $N=3$ results
(\ref{medopp2}) and (\ref{me3wg}) is that it might make sense to
search for the general $N$ by $N$ metrics $\Theta^{(N)}(\gamma)$ in
the form which remains complex symmetric with respect to its
second diagonal.

\subsection{Nine-parametric ansatz}

The first nontrivial $N=4$ CBH Hamiltonian
 \be
 H^{(4)}_{(CBH)}(\gamma)= \left[ \begin {array}{cccc}
 -3\,i\gamma&\sqrt{3}&0&0
 \\\noalign{\medskip}\sqrt {3}&-i\gamma&2&0
 \\\noalign{\medskip}0&2&i\gamma&\sqrt {3}
 \\\noalign{\medskip}0&0&\sqrt {3}&3\,i\gamma\end {array}
 \right]\,
 \label{tripa}
 \ee
yields the bound-state energies $E_{\pm,\pm} = \pm (2 \pm
1)\,\sqrt{1-\gamma^2}$ which are real inside a finite interval
bounded by its exceptional-point boundaries $\gamma^{(EP)}_\pm =\pm
1$. We must
find now such a ``standard'' Hilbert space ${\cal H}^{(S)}$ in which the
evolution generated by our ``closed-system'' Hamiltonian
$H^{(4)}_{(CBH)}(\gamma)$
would be unitary.
Thus, we must find at least one invertible
and Hermitian matrix
 \be
 \Theta^{(4)}_{(CBH)}(\gamma)=
 \left[ \begin {array}{cccc} 1&{\it {\beta}}+ix&
 {\it {\delta}}+iz&{\it {\kappa}}
 +iy\\\noalign{\medskip}{\it {\beta}}-ix&{\it {\rho}}&{\it
 {\tau}}+iu&{\it
 {\delta}}+iz
 \\
 \noalign{\medskip}{\it {\delta}}-iz&{\it {\tau}}-iu&
 {\it {\rho}}&{\it {\beta}}+ix
 \\
 \noalign{\medskip}{\it {\kappa}}-iy&{\it {\delta}}-iz&{\it
 {\beta}}-ix&1
 \end {array} \right]\,
 \label{anzam4g}
 \ee
compatible with the Dieudonn\'{e}'s hidden-Hermiticity constraint
(\ref{requi}) in an interval of $\gamma \in (0,\gamma^{(4)}_{max})$.

\begin{lemma}
Up to an arbitrary overall multiplication factor, ansatz
(\ref{anzam4g}) leads to the fully general CBH solution of
Dieudonn\'e's Eq.~(\ref{requi}) at $N=4$.
\end{lemma}

\begin{proof}
Elementary linear algebra converts Eq.~(\ref{requi})  into the set
of definitions of the real diagonal element
 $$
{\rho}=1/3\, \left( 4\,{{\it \gamma}}^{2}\sqrt {3}
+\sqrt {3}+2\,{\it
{\delta}} \right) \sqrt {3} = 1+4\,{{\it \gamma}}^{2}+2{\it
{\delta}}/\sqrt {3}\,
 $$
as well as of all of the imaginary parts of the matrix elements,
 $$
x=\sqrt {3}{\it \gamma}\,,\ \ \ \
z=2\,{\it \gamma}\,{\it {\beta}}\,,\ \ \ \
 u=2\,{{\it \gamma}}^{3}+2\,{\it \gamma}
 +{\it \gamma}\,{\it {\delta}}/\sqrt {3}\,,
 \ \ \ \
 y=
 {\it \gamma}\, \left( 2\,{{\it \gamma}}^{2}+
\sqrt {3}\,{\it {\delta}}
 \right)\,.
 $$
The last Dieudonn\'e's constraint
 $$
 {\tau}= \left( 4\,{{\it \gamma}}^{2}{\it {\beta}}
 +2\,{\it {\beta}}+\sqrt {3}{\it {\kappa}}
 \right)/ \sqrt {3}\,
 $$
defines the last redundant real parameter so that
Eq.~(\ref{anzam4g}) acquires the status of the fully general
three-parametric  definition of the candidates
$\Theta^{(4)}(\gamma)=\Theta^{(4)}(\gamma,\beta,\delta,\kappa)$ for
the metric.
\end{proof}

 \noindent
The latter observation can be extrapolated to any matrix dimension
$N$.

\begin{thm}
For the matrix form of the $c=0$ CBS Hamiltonians (\ref{Uwemo}) of
any dimension $N$ the Dieudonn'{e}'s Eq.~(\ref{requi}) defines,
recurrently, all of the matrix elements of all of the eligible
(though not yet necessarily positive definite) Hamiltonian-dependent
metric candidates $\Theta=\Theta(H)$ in terms of the freely variable
$N-$plet of the real parts of the elements in the first row.
\end{thm}

\begin{proof}
For the proof it is sufficient to realize that as long as our
Hamiltonians are tridiagonal, the Dieudonn\'e's Eq.~(\ref{requi})
can be rearranged, row-wise, as a set of recurrences. A fully
detailed account of this idea and the explicit form of the
arrangement of the recurrences in an analogous real-matrix case may
be found in Ref.~\cite{rekur}.
\end{proof}

\subsection{Zero-parametric ansatz for the metric at $N=4$}

When we compare the $(N-1)-$parametric candidates for the metrics as
obtained at $N=2$ [cf. Eq.~(\ref{medopp2})] and at $N=3$ [cf.
Eq.~(\ref{me3wg})] we notice that they are positive definite at the
vanishing parameters (cf. $\beta\to 0$ at $N=2$ in
Eq~(\ref{alleig}), or $\beta\to 0$ and $\delta \to 0$ at $N=3$ in
Eq~(\ref{someig}). Such a regularity feature survives also at $N=4$:
for the proof we merely select ${\beta}={\delta}={\kappa}=0$ and
check.

\begin{lemma}
For ${\beta}={\delta}={\kappa}=0$ and for the sufficiently small
$\gamma$, all of the four eigenvalues of matrix (\ref{anzam4g}) will
be positive and will have the following leading-order form
 \be
 \label{resu}
 \theta_{\pm,1}=1\pm \gamma+O \left( {{\it \gamma}}^{2}
 \right)\,,
 \ \ \ \  \theta_{\pm,2}=1\pm 3\,\gamma+O \left( {{\it \gamma}}^{2}
 \right)\,.
 \ee
\end{lemma}

\begin{proof}
The leading-order minimization (\ref{resu}) of the physical
Hilbert-space anisotropy finds its immediate inspiration in the
general recipe of Ref.~\cite{lotor}. For our particular model the
formula can be obtained directly from the reduced
${\beta}={\delta}={\kappa}=0$ version of Eq.~(\ref{anzam4g}),
 \be
 \Theta^{(4)}(\gamma,0,0,0)=
 \left[ \begin {array}{cccc}
 1&i\sqrt {3}{\it \gamma}&0&2\,i{{\it \gamma}}^{3}\\
 \noalign{\medskip}-i\sqrt {3}{\it \gamma}&4\,{{\it \gamma}}^{2}+1
 &2\,i{\it \gamma}\,
  \left( {{\it \gamma}}^{2}+1 \right) &0\\
 \noalign{\medskip}0&-2\,i{\it \gamma}\, \left( {{\it \gamma}}^{2}+1
 \right)
  &4\,{{\it \gamma}}^{2}+1&i\sqrt {3}{\it \gamma}\\
\noalign{\medskip}-2\,i{{\it \gamma}}^{3}&0& -i\sqrt {3}{\it
\gamma}&1
\end {array} \right]\,.
 \label{beznich}
 \ee
Elementary algebra yields the four eigenvalues $\theta_j$ in the
respective closed forms
 $$
1+2\,{{\it \gamma}}^{2}-{\it \gamma} \pm 2\,\sqrt {{{\it
\gamma}}^{2}+2\,{{\it \gamma}}^{4 }-{{\it \gamma}}^{3}+{{\it
\gamma}}^{6}-2\,{{\it \gamma}}^{5}}\,,
 $$
and
 $$
 1+{\it \gamma}+2\,{{\it \gamma}}^{2}\pm 2\,\sqrt {{{\it
 \gamma}}^{2}
 +{{\it \gamma}}^{3}+2
\,{{\it \gamma}}^{4}+{{\it \gamma}}^{6}+2\,{{\it \gamma}}^{5}}
 $$
which can be Taylor-expanded,
 $$
 \theta_1=1+{\it \gamma}+{{\it \gamma}}^{2}
 +O \left( {{\it \gamma}}^{3}
 \right)\,,\ \ \ \ %
 \theta_2=1-3\,{\it \gamma}+3\,{{\it \gamma}}^{2}
+O \left( {{\it \gamma}}^{3 } \right) \,,
 $$
 $$
 \theta_3=1+3\,{\it \gamma}+3\,{{\it \gamma}}^{2}
+O \left( {{\it \gamma}}^{3 } \right) \,,\ \ \ \ %
 \theta_4=1-{\it \gamma}+{{\it \gamma}}^{2}
+O \left( {{\it \gamma}}^{3}
 \right)\,.
$$
\end{proof}
 \noindent
The low-order terms may be checked to stay unchanged when we omit
the higher-order terms directly from the matrix $\Theta$. After such
a simplification the secular equation yields the same leading-order
roots in the closed linear form (\ref{resu}) more quickly. Such a
simplification proves useful at the higher $N$.

At $N=4$ the confirmation of this trick may rely on the split of
eigenvalues $\theta_j=1+s_j$ and on the search of the roots $s_j$ of
the rearranged secular equation
 \be
 0={s}^{4}-8\,{s}^{3}{{\it \gamma}}^{2}+ \left( -8\,{{\it
 \gamma}}^{6}-10\,{{\it \gamma}}^{2}+8\,{{\it \gamma}}^{4} \right)
 {s}^{2}+ \left( 32\,{{\it \gamma}}^{8}+ 24\,{{\it \gamma}}^{4}
 \right) r+24\,{{\it \gamma}}^{6}+9\,{{\it \gamma}}^{4}+40\, {{\it
 \gamma}}^{8}-32\,{{\it \gamma}}^{10}+16\,{{\it \gamma}}^{12}\,.
 \label{secuture}
 \ee
Under the assumption that  the  roots $s_j$ are all small for small
$\gamma$s, we may now order the coefficients and keep just the
leading-order terms. Eq.~(\ref{secuture}) then degenerates to the
solvable problem $0=s^4-10\,\gamma^2s^2+9\gamma^4$ possessing the
same leading-order roots as above, with $s_{\pm,1} =\pm \gamma +
\ldots$ and $s_{\pm,2} =\pm 3\gamma + \ldots$.

On the basis of such an experience one could recognize a
re-emergence of angular-momentum matrices $L_{x,y,z}^{(N)}$ and
prove the following general$-N$ result.

\begin{thm}
\label{[3]} Up to the first order in $\gamma$ one of the simplest
physical CBH metrics may be given, at any $N$, the closed form
 \be
 \Theta^{(N)}(\gamma)=I- 2\,\gamma\, L_y^{(N)}
 +{\cal O}(\gamma^2)\,.
 \label{obecm}
 \ee
\end{thm}

\begin{proof}
It is sufficient to recall Eq.~(\ref{3wg}) and the commutation
relations in $su(2)$.
\end{proof}

\begin{cor}
\label{[5]}
The first-order eigenvalues of metric (\ref{obecm})
form the equidistant set
 $$
 \{\theta_j\}=
 \{
  1-2\,J\,\gamma\,,\,1-2\,(J-1)\,\gamma\,,\,\ldots\,,\,1-2\,\gamma\,,\,
 1, 1+2\,\gamma\,,\,
 1+4\,\gamma\,,\, \ldots\,,\, 1+2\,J\,\gamma\,
 \}
 $$
at odd $N=2J+1$, and the equidistant set
 $$
 \{\theta_j\}=
 \{
  1-(2\,J+1)\,\gamma\,,\,1-(2\,J-1)\,\gamma\,,\,\ldots\,,\,1-\gamma\,,\,
 1+\gamma\,,\,
 1+3\,\gamma\,,\, \ldots\,,\, 1+(2\,J+1)\,\gamma\,
 \}
 $$
at even $N=2J+2$.
\end{cor}

 \noindent
The higher-power corrections will start playing a role when the
strength $\gamma$ of the non-Hermiticity ceases to be small. Still,
the smallest eigenvalue $1-(N-1)\gamma + c\,\gamma^2+\ldots$ will
vanish, in the leading-order approximation,  at $\gamma_{critical}
\approx 1/(N-1)$. At this boundary the metric will lose its
invertibility and positivity.

The estimate of $\gamma_{critical}$ decreases with $N$. Its
second-order amendment becomes $c-$dependent so that once we return
to the $N=4$ example (in which $c$ is positive --  cf. the
Taylor-series formulae in the proof of Theorem \ref{[3]}), the
incorporation of the second order correction makes the amended
estimate of $\gamma_{critical}$ slightly larger. Incidentally, for
our $N=4$ matrix (\ref{beznich}) the non-approximative, exact value
of $\gamma_{cricical}=1/\sqrt{2}\ $ happens to be even larger.

\subsection{The $N=4$ metric covering the whole range of $\gamma$}

With trivial $\beta=\kappa=0$ and a tentative choice of $\delta =
\delta(\gamma)$ one obtains vanishing $z=\tau= 0$ and nontrivial
$y$, $u$ and $\rho$. All of the matrix elements of $\Theta$ appear
to be nonzero and strictly real or purely imaginary. In a trial and
error manner we verified that for $\delta(\gamma) =
-\sqrt{3}(\gamma+\nu\,\gamma^3)/(\nu +1)$ with $\nu=1$, $\nu=2$,
$\nu=3$ or $\nu=4$ the metric remains regular and positive along the
whole physically relevant interval of $\gamma\in (0,1)$. Even at the
smallest $\nu=1$ the graphical illustration of the collapse of the
invertibility in the EP limit $\gamma\to 1$ displayed in
Fig.~\ref{uba1} is persuasive. Omitting the fourth root (such that
$\lim_{\gamma \to 1}\theta_4(\gamma)=8$) the picture shows that all
of the other eigenvalues of the metric move to zero at a very
different rate with $\gamma$ reaching its EP maximum $\gamma=1$.

\begin{figure}[h]                    
\begin{center}                         
\epsfig{file=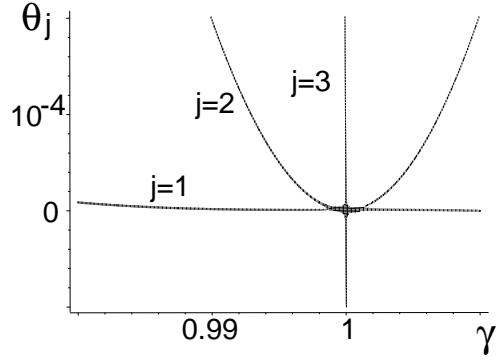,angle=270,width=0.4\textwidth}
\end{center}    
\vspace{2mm} \caption{The decrease of the three lowest eigenvalues
$\theta_j(\gamma)$ of the $N=4$ matrix (\ref{anzam4g}) to zero for
the growth of $\gamma$ to its maximal physical value
$\gamma^{(EP)}=1$. The fourth root is much larger and lies out of
the picture.
 \label{uba1}
 }
\end{figure}

In the domain of
small $\gamma$ and at $\delta=0$, the $\gamma-$proportionality
coefficients (-3,-1,1,3) as given by Theorem \ref{[3]}
and its Corollary \ref{[5]} get
modified and replaced by the
less compact set of
quantities $-3/2\pm\sqrt {3}$ and $1/2\pm\sqrt {7}$, i.e.,
numerically, -3.23,-2.15, 0.23, 3.15, obtained as the roots of the
leading-order secular polynomial
$s^4+(2)\,s^3+(-21/2)\,s^2+(-39/2)\,s+81/16$.


\section{Construction of metrics at $N\geq 5$\label{vvkol}}

\subsection{Physical Hilbert space at
$N=5$\label{5kol}}

In the domain of very small $\gamma$ the process of the
Hermitization of the CBH matrix Hamiltonians may be based on Theorem
\ref{[3]} at any $N$. Nevertheless, as long as the estimated
$\gamma_{critical} \approx 1/(N-1)$ decreases quickly with the
growth of $N$, the range of applicability of the
linear approximation becomes more and more restricted. The
higher-precision constructions become of an enhanced interest at
$N\geq 5$.

\subsubsection{General four-parametric metric}

For
 \be
 H^{(5)}_{(CBH)}(\gamma)=\left[
 \begin {array}{ccccc} -4\,i\gamma&2&0&0&0
 \\
 \noalign{\medskip}2&-2\,i\gamma&\sqrt {6}&0&0
 \\
 \noalign{\medskip}0&\sqrt{6}&0&\sqrt {6}&0\\\noalign{\medskip}0&0&
 \sqrt {6}&2\,i\gamma&2\\\noalign{\medskip}0&0&0 &2&4 \,i\gamma\end
 {array} \right]\,
 \label{petpa}
 \ee
we may consider the general Hermitian candidate for the metric with
normalization
 $
\Theta^{(5)}_{1,1}(\gamma) =1$,
 \be
 \Theta^{(5)}(\gamma)=
  \left[ \begin {array}{ccccc} {\it 1}&{\it {X}}+ix&{\it {Y}}+iy&{\it
  {Z}
 }+iz&{\it {W}}+iw\\\noalign{\medskip}{\it {X}}-ix&{\it {R}}&{\it {U}}+iu&{
 \it {V}}+iv&{\it {Z}}+iz\\\noalign{\medskip}{\it {Y}}-iy&{\it {U}}-iu&{
 \it {T}}&{\it {U}}+iu&{\it {Y}}+iy\\\noalign{\medskip}{\it {Z}}-iz&{\it {V}
 }-iv&{\it {U}}-iu&{\it {R}}&{\it {X}}+ix\\\noalign{\medskip}{\it {W}}-iw&{
  \it {Z}}-iz&{\it {Y}}-iy&{\it {X}}-ix&{\it 1}\end {array} \right]\,.
 \label{petia}
 \ee
Its insertion in Eq.~(\ref{requi}) yields the sequence of
definitions of the four imaginary components of the first-row
elements,
 $$
 x=2\,{\it \gamma}\,,\ \ y=\sqrt {6}{\it \gamma}\,{\it {X}}\,,
 \ \ \ \ z=1/3\,{\it \gamma}\, \left( 3\,{\it {Y}}+2\,{{\it
 \gamma}}^{2}\sqrt
 {6}
 \right) \sqrt {6}\,,
 \ \ \ \ w=v+{\it \gamma}\,{\it {Z}}-3\,{\it {X}}\,{\it \gamma}
 $$
as well as the two imaginary components of the second-row elements,
 $$
 u=1/6\,{\it \gamma}\, \left( 6+12\,{{\it \gamma}}^{2}+\sqrt {6}{\it
 {Y}}
 \right) \sqrt {6}\,,
 \ \ \ \
 v=4\,{{\it \gamma}}^{3}{\it {X}}+3\,{\it {X}}\,{\it \gamma}+{\it
 \gamma}\,{\it {Z}}\,.
 $$
The real parts to be defined are the two off-diagonal items
 $$
 {\it {U}}=2\,{{\it \gamma}}^{2}\sqrt {6}{\it {X}}+1/2\,\sqrt {6}{\it
 {X}}+1/2 \,\sqrt {6}{\it {Z}}\,,
 \ \ \ \ {\it {V}}={\it {W}}+{{\it \gamma}}^{2}\sqrt {6}{\it
 {Y}}+4\,{{\it
 \gamma}}^{4}+1/ 2\,\sqrt {6}{\it {Y}}
 $$
and their two diagonal partners
 $$
 {\it {R}}=1+1/2\,\sqrt {6}{\it {Y}}+6\,{{\it \gamma}}^{2}\,,
 \ \ \ \ {\it {T}}=1/6\, \left( 4\,{\it {Y}}+8\,{{\it
 \gamma}}^{2}\sqrt
 {6}+8\,{{\it \gamma}}^{4}\sqrt {6}+8\,{{\it \gamma}}^{2}{\it
 {Y}}+\sqrt {6}+\sqrt {6}{\it {W} } \right) \sqrt {6}\,.
 $$
The resulting
$\Theta^{(5)}(\gamma)=\Theta^{(5)}(\gamma,{X},{Y},{Z},{W})$ is a
fairly compact candidate for the general four-parametric metric.

\subsubsection{The requirement of positivity}

The key technical obstacle arises when one wishes to specify the
exact boundaries of the physical domain of the parameters for which
the metric remains well defined, i.e., invertible and positive
definite. For our present purposes we only need to know the metric
at the small $\gamma$. In a way
indicated by Corollary \ref{[5]}, the eigenvalues of the metric are
then revealed positive and equidistant.

At $N=5$  the approximate metric candidate which is linear in
$\gamma$ reads
 $$
 \Theta^{(5)}_0(\gamma)=
 \left[ \begin {array}{ccccc} 1&2\,i{\it \gamma}&0&0&0
 \\
 \noalign{\medskip}-2\,i{\it \gamma}&1&
 i{\it \gamma}\,\sqrt {6}&0&0\\\noalign{\medskip}0&-i{\it
\gamma}\,\sqrt {6}&1&i{\it \gamma}\,\sqrt
{6}&0\\\noalign{\medskip}0&0&-i{\it \gamma}\,\sqrt {6}&1&2\,i{\it
 \gamma}\\\noalign{\medskip}0&0&0&-2\,i{\it \gamma}&1
 \end {array} \right]\,.
 $$
The quintuplet of its exact eigenvalues $\theta_j=1+r_j$ coincides with
the roots of the exact secular polynomial  which gets completely
factorized,
 $$
 0=\left( \theta-1 \right)
 \left( \theta-1+2\,{\it \gamma} \right)  \left( \theta-1-2\,{
\it \gamma} \right)  \left( \theta-1-4\,{\it \gamma} \right)  \left(
\theta-1+4\,{\it \gamma } \right)\,.
 $$
For the proof it is sufficient to turn attention to the deviations
$r_j$ from the unit value. Then, the {\em exact\,} secular equation
is again easily derived.
As long as its explicit form becomes rather lengthy (containing as
many as 22 terms), the detailed discussion of the mutual dependence
of its parameters and roots would be a formidable task. Fortunately,
once one omits all of the higher order corrections, the reduced
equation reads ${r}^{5}-20\,{{\it \gamma}}^{2}{r}^{3}+64\,r{{\it
\gamma}}^{4}=0$ and remains solvable easily yielding the roots $r$
proportional to $\gamma$ (i.e., small).

\subsubsection{Two-parametric subfamily of metrics}

In the  models with $N=2$, 3 and 4 we saw that an
important simplification resulted from the purely
imaginary choice of $\Theta_{1,2}$. One of the consequences was the
purely real form of  $\Theta_{1,3}$. Let us now try to generalize
this experience and introduce, tentatively, a chessboard-inspired
ansatz with, in general, ${\rm Im\ }\Theta_{i,j}=0$ for $i+j=$ even
and ${\rm Re\ }\Theta_{i,j}=0$ for $i+j=$ odd. At $N=5$ it reads
 $$
 \Theta^{(5)}(y,w)=
 \left[
 \begin {array}{ccccc} 1&ix&{\it y}&iz&{\it w}
 \\
 \noalign{\medskip}-ix&{\it {\rho}}&iu&{\it v}&iz
 \\
 \noalign{\medskip}{
\it y}&-iu&{\it {\tau}}&iu&{\it y}
\\
\noalign{\medskip}-iz&{\it v}&-iu&{
\it {\rho}}&ix
\\
\noalign{\medskip}{\it w}&-iz&{\it y}&-ix&1\end {array}
 \right]\,.
 $$
Here, in the light of Eq.~(\ref{requi}) we have $x=2\,{\it \gamma}$,
$z=\sqrt {6}{\it y}\,{\it \gamma}+4\,{{\it \gamma}}^{3}$, $u={\it
y}\,{\it \gamma}+{\it \gamma}\,\sqrt {6}+2\,\sqrt {6}{{\it
\gamma}}^{3}$ while $v=\sqrt {6}{\it y}\,{{\it \gamma}}^{2}+4\,{{\it
\gamma}}^{4}+1/2\,{\it y}\, \sqrt {6}+{\it w} $. On the main
diagonal we get ${\rho}=6\,{{\it \gamma}}^{2}+1+1/2\,{\it y}\,\sqrt
{6}$ and ${\tau}=1+{\it w}+\left( 4\,{\it y}+8\,{\it y}\,{{\it
\gamma}}^{2}\right)/ \sqrt {6} +8\,{{ \it \gamma}}^{2}+8\,{{\it
\gamma}}^{4}
 $.
The resulting remarkable pattern of the distribution of the powers
of $\gamma$ is best visible when we choose trivial $y=w=0$ and get
the maximally elementary metric without free
parameters,
 \be
 \Theta^{(5)}=
 \left[ \begin {array}{ccccc} 1&2\,i{\it \gamma}&0&4\,i{{\it
 \gamma}}^{3}&0
 \\
 \noalign{\medskip}-2\,i{\it \gamma}&6\,{{\it \gamma}}^{2}+1
 &i{\it \gamma}\,\sqrt {6} \left( 1+2\,{{\it \gamma}}^{2} \right)
 &4\,{{\it \gamma}}^{4}&4\,
 i{{\it \gamma}}^{3}\\\noalign{\medskip}0&-i{\it \gamma}\,\sqrt {6}
 \left( 1+2\,{{\it \gamma}}^{2} \right) &1+8\,{{ \it
 \gamma}}^{2}+8\,{{\it \gamma}}^{4}&i{\it \gamma}\,\sqrt {6} \left(
 1+2\,{{\it \gamma}}^{2} \right) &0
 \\
 \noalign{\medskip}-4\,i{{\it \gamma}}^{3}
 &4\,{{\it \gamma}}^{4}&
 -i{\it \gamma}\,\sqrt {6} \left( 1+2\,{{\it \gamma}}^{2} \right)
 &6\,{{\it \gamma}}^{
 2}+1&2\,i{\it \gamma}\\\noalign{\medskip}0&-4\,i{{\it
 \gamma}}^{3}&0&-2\,i{ \it \gamma}&1\end {array} \right]\,.
 \label{gived}
 \ee
This
matrix is positive definite (i.e., eligible metric) in the
reasonably large interval of $\gamma\in (0,\gamma_{\max}^{(5)})$
where $\gamma_{\max}^{(5)}=1/2\,\sqrt {\sqrt {5}-1} \approx
0.5558929700$. The $\gamma-$dependence of the five eigenvalues
$\theta_j$ is shown in Fig.~\ref{uba2a}.

\begin{figure}[h]                    
\begin{center}                         
\epsfig{file=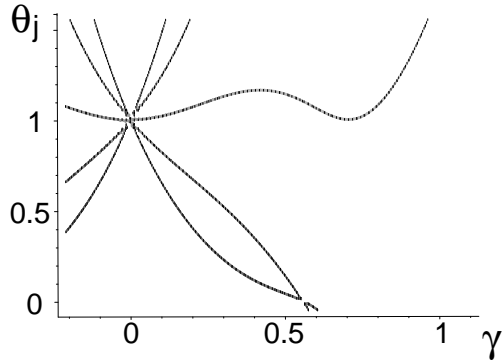,angle=270,width=0.4\textwidth}
\end{center}    
\vspace{2mm} \caption{The $\gamma-$dependence of the eigenvalues of
the $N=5$ matrix (\ref{gived}) and their positivity in the interval
of $\gamma \in \left [0,\gamma_{\max}\right )$ with
$\gamma_{\max}=1/2\,\sqrt {\sqrt {5}-1}$.
 \label{uba2a}
 }
\end{figure}

\subsection{Physical Hilbert space at
$N=6$\label{mol}}

Our last explicit chessboard-inspired ansatz reads
 \be
 \Theta^{(6)}= \left[ \begin {array}{cccccc}
  1&{\rm i\,}x&{\it {y}}&{\rm i\,}z&{\it {w}}&{\rm i\,}{\mu}
\\
\noalign{\medskip}-{\rm i\,}x &{\it {\rho}}&{\rm i\,}u&{\it
{v}}&{\rm i\,}{\it {\zeta}}&{\it {w}}
\\
\noalign{\medskip}{\it {y}}&-{\rm i\,}u
&{\it {\tau}}&{\rm i\,}{\sigma}&{\it {v}}&{\rm i\,}z
\\
\noalign{\medskip}-{\rm i\,}z&{\it {v}}
&-{\rm i\,}{\sigma}&{\it {\tau}}&{\rm i\,}u&{\it {y}}
\\
\noalign{\medskip}{\it {w}}&-{\rm i\,}{\it {\zeta}}
&{\it {v}}&-{\rm i\,}u&{\it {\rho}}&{\rm i\,}x
\\\noalign{\medskip}-{\rm i\,}{\mu}&{\it {w}
}&-{\rm i\,}z&{\it {y}}&-{\rm i\,}x&1\end {array}
 \right].
 \label{fullfle}
 \ee
Its imaginary matrix elements are all defined by formulae
 $x=\sqrt {5}{\it \gamma}$,
 $z=2\,\sqrt {5}{{\it \gamma}}^{3}\sqrt {2}+3\,{\it {y}}\,{\it
 \gamma}$
 and
 ${\mu}=16\,{{\it \gamma}}^{5}
 +2\,{{\it \gamma}}^{3}\sqrt {2}\sqrt {5}{\it {y}}+\sqrt {
5}{\it {w}}\,{\it \gamma} $
for the first row and by formulae
 $u=2\,{\it \gamma}\,\sqrt {2}
 -6/5\,\sqrt {5}{\it {y}}\,{\it \gamma}+3/5\, \left( 2
\,\sqrt {5}{{\it \gamma}}^{3}\sqrt {2}+3\,{\it {y}}\,{\it \gamma}
\right) \sqrt {5} $,
  ${\zeta}=
  2/5\,\sqrt {2} \left( 2\,\sqrt {5}{{\it \gamma}}^{3}\sqrt {2}+3\,{\it
  {y}}
\,{\it \gamma} \right) \sqrt {5}+3/5\,\sqrt {5}{\it {w}}\,{\it
\gamma}+16\,{{ \it \gamma}}^{5}+2\,{{\it \gamma}}^{3}\sqrt {2}\sqrt
{5}{\it {y}} $
 and
 ${\sigma}=3\,{\it \gamma}
 -6/5\,\sqrt {2}\sqrt {5}{\it {y}}\,{\it \gamma}+3/5\,\sqrt {2}
 \left( 2\,\sqrt {5}{{\it \gamma}}^{3}\sqrt {2}+3\,{\it {y}}\,{\it
 \gamma}
 \right) \sqrt {5}-2/5\,{{\it \gamma}}^{2}\sqrt {2} \left( 2\,\sqrt
 {5}{{
\it \gamma}}^{3}\sqrt {2}+3\,{\it {y}}\,{\it \gamma} \right) \sqrt
{5}+1/5\, \sqrt {5}{\it {w}}\,{\it \gamma}+16\,{{\it
\gamma}}^{5}+2\,{{\it \gamma}}^{3}\sqrt {2}\sqrt {5}{\it {y}} $
for the second and the third row. As long as $y$ and $w$
are kept as independent variables, we only have to define
the remaining real matrix elements
 ${\rho}=8\,{{\it \gamma}}^{2}+1+2/5\,\sqrt {5}{\it {y}}\,\sqrt {2}$
 and
 ${\tau}=3/5\,\sqrt {5}{\it {y}}\,\sqrt {2}+12\,{{\it
 \gamma}}^{2}-6/5\,\sqrt
 {2}
\sqrt {5}{\it {y}}\,{{\it \gamma}}^{2}+6/5\,\sqrt {2}{\it \gamma}\,
\left( 2\, \sqrt {5}{{\it \gamma}}^{3}\sqrt {2}+3\,{\it {y}}\,{\it
\gamma} \right) \sqrt { 5}+1+3/5\,{\it {w}}\,\sqrt {5} $
on the main diagonal, and, last but not least,
 ${v}=4/5\,{\it \gamma}\,
  \left( 2\,\sqrt {5}{{\it \gamma}}^{3}\sqrt {2}+3\,{\it {y}}
\,{\it \gamma} \right) \sqrt {5}+3/5\,{\it {y}}\,\sqrt
{5}+2/5\,\sqrt {5}{ \it {w}}\,\sqrt {2} $.

These results confirm the expectations and extrapolation hypotheses.
Once we accept the most natural simplification {${y}={w}=0$} and
once we keep just the terms which are linear in $\gamma$ we get
again the tridiagonal metric (\ref{obecm})
with the factorized secular equation,
 $$
 0=\left( \theta-1 \right)  \left( \theta-1+2\,{\it \gamma} \right)
   \left( \theta-1-2\,{
\it \gamma} \right)  \left( \theta-1-4\,{\it \gamma} \right)  \left(
\theta-1+4\,{\it \gamma } \right)\,.
 $$
After one incorporates the higher-order corrections we get the
parameter-free version of the metric. The simplification of its
matrix elements included
 %
 $z=2\,\sqrt {5}{{\it \gamma}}^{3}\sqrt {2}$,
 $\mu=16\,{{\it \gamma}}^{5}
$,
$u=2\,{\it \gamma}\,\sqrt {2}+6\,{{\it \gamma}}^{3}\sqrt {2}$,
 ${\zeta}=8\,{{\it \gamma}}^{3}+16\,{{\it \gamma}}^{5}
$
and
 $\sigma=3\,{\it \gamma}+12\,{{\it \gamma}}^{3}+8\,{{\it \gamma}}^{5}
$
as well as
 ${\rho}=1+8\,{{\it \gamma}}^{2}$,
 ${\tau}=1+12\,{{\it \gamma}}^{2}+24\,{{\it \gamma}}^{4}
$
and, finally,
 $v=8\,{{\it \gamma}}^{4}\sqrt {2}
$. This enables us to display here the whole matrix, with its
symmetry-determined matrix elements omitted,
 \be
 \Theta^{(6)}_{(CBH)}= \left[ \begin {array}{ccccc}
  1&i\sqrt {5}{\it \gamma}&0&\ldots&16\,i{{\it
\gamma}}^{5}
\\
\noalign{\medskip}-i \sqrt {5}{\it \gamma}&1+8\,{{\it \gamma}}^{2}&
2\,i\,\sqrt {2}\, {\it \gamma}\,\left(1+ 3\,{{\it \gamma}}^{2}
\right) &\ldots &0
\\
\noalign{\medskip}0&-
2\,i\,\sqrt {2}\, {\it \gamma}\,\left(1+ 3\,{{\it \gamma}}^{2}
\right) &1+12\,{{\it \gamma}}^{2}+24\,{{\it
\gamma}}^{4}&\ldots&2\,i\sqrt {10}{{\it \gamma}}^{3}
\\
\noalign{\medskip}-2\,i\sqrt {10}{{\it \gamma}}^{3}&8\,\sqrt
{2}\,{{\it \gamma} }^{4}&-i\,{\it \gamma}\, \left( 3+12\,{{\it
\gamma}}^{2}+8\,{{\it \gamma}}^ {4} \right) &\ldots &0
\\
\noalign{\medskip}0&-8\,i\,{{\it \gamma}}^{3} \left( 1+2\,{{\it
\gamma}}^{2}
 \right) &8\,{{\it \gamma}}^{4}\sqrt {2}&\ldots&i\sqrt {5}{ \it
\gamma}
\\
\noalign{\medskip}-16\,i{{\it \gamma}}^{5}&0&-2\,i\,\sqrt
{10}\,{{\it \gamma}}^{3}&\ldots&1\end {array} \right]\,.
 \label{lasta}
 \ee
From the exact value of its smallest eigenvalue, incidentally, we
managed to deduce the exact value of the domain-boundary quantity
$\gamma^{(6)}_{\max}=1/2$. Moreover, we also managed to evaluate the
second-order corrections to the eigenvalues of the metric
$\theta_j=1+A_j\,\gamma+B_j\,\gamma^2+\ldots$. This result,
summarized in Table \ref{pexp4}, may be read as a strong
encouragement of extrapolations towards $N>6$.

\begin{table}[h]
\caption{Coefficients in the approximation
$\theta^{(6)}_j(\gamma)=1+A_j\,\gamma+B_j\,\gamma^2+\ldots$ }
\label{pexp4}

\vspace{2mm}

\centering
\begin{tabular}{||c||c|c||}
\hline \hline
   $j$  & $A_j$ &
     $B_j$
    \\
 \hline \hline
 1& 5  & 10\\
 2&  3  & 6\\
 3&  1 & 4 \\
 4& -1 & 4 \\
 5&  -3 & 6 \\
 6  &  -5& 10\\
 \hline
 \hline
\end{tabular}
\end{table}

 \noindent
Presumably, also many other features of the $N=6$ CBH model may be
expected to find their analogues at the general matrix dimensions
$N$. Among the clearest candidates, tendencies and patterns of
possible extrapolations let us mention the last few hypotheses which
could help us to reduce, in the future, the recurrences for the
matrices $\Theta^{(N)}$ with polynomial entries to the recurrences
for the mere arrays of the coefficients.

\begin{itemize}

\item
At any $N$, the chessboard-inspired complex ansatz for
$\Theta^{(N)}(\gamma)$ may be required to be a strictly real matrix
after a formal analytic-continuation replacement of $\gamma$ by
${\rm i} \alpha$ with real $\alpha$;

\item
the powers $\gamma^{K_n}$  entering the matrix element
$\Theta^{(N)}_{i,j}(\gamma)$ may be  conjectured to be limited by
the following empirical rules:

\begin{enumerate}

\item
$K_n$=even iff $i+j=$ even; $K_n$=odd iff $i+j=$ odd;

\item
$\min K_n = |i-j|$;
$\max K_n = \min(|i+j-2|,|i+j-2N-2 |)$.

\end{enumerate}

\end{itemize}

\section{Summary\label{summat}}

Quantum theory offers a counterintuitive picture of reality. One has
to replace, e.g., the energy of a classical system by an operator.
In the most common unitary-evolution scenario such an operator
(i.e., Hamiltonian) must be self-adjoint,
$\mathfrak{h}=\mathfrak{h}^\dagger$. Once we admit a ``hidden''
interaction with environment, it may cause the loss of the
self-adjointness of the ``effective'' Hamiltonian, $H \neq
H^\dagger$ \cite{Nimrod}. In parallel, the spectrum becomes complex
and, due to the possible losses or gains from the environment, the
evolution ceases to be unitary.

For a long time it escaped the attention of physicists that there
also exists a fairly large family of quantum systems in which the
Hamiltonians are admitted non-Hermitian but still, the system
remains closed, exhibiting no interaction with an ``environment''.
The evolution is unitary, in an apparent contradiction with the
Stone's theorem. Fortunately, the paradox results from a
misunderstanding: the ``false'' non-self-adjointness is detected in
an ill-chosen Hilbert space ${\cal H}^{(F)}$.

Such an innovative use of non-Hermitian generators $H$ of the
unitary evolution found its most persuasive success in nuclear
physics \cite{Geyer} or in condensed-matter physics \cite{Dyson}.
The more ambitious theoretical implementations of the idea were
pursued in perturbation theory \cite{BG}, in certain relativistic
\cite{aliKG} and supersymmetric \cite{susy} extensions of quantum
mechanics plus, perhaps, in quantum cosmology \cite{Bang} and in
quantum theory of catastrophes \cite{catast}. In all of these
applications, Dieudonn\'{e} relation~(\ref{requi}) appeared to
connect a given non-Hermitian observable Hamiltonian $H$ with all of
its admissible Hermitizations, i.e., with all of the eligible
physical Hilbert-space metrics $\Theta=\Theta(H)$. For our present
Bose-Hubbard family (\ref{Uwemo}) of quantum Hamiltonians, in
particular, we were able to guarantee the stable, unitary evolution
of the system via the construction  of the operator
$\Theta=\Theta(H)$ at small $\gamma$. For each representation (i.e.,
matrix dimension $N$) we recommended the direct solution of
Dieudonn\'{e}'s Eq.~(\ref{requi}) and we showed that it is feasible.

For the first nontrivial matrix dimension $N=3$ we admit that the
purely algebraic part of the task already looks rather complicated.
Still, a suitable amendment of the approach made the construction
feasible. The essence of the simplification lies in the ansatz for
$\Theta(H)$ with symmetry with respect to the second diagonal. We
revealed that such an assumption leads to the fully general
three-parametric family of the candidates for the metric at $N=3$,
and that it might open the way towards the study of models with
general $N$.

At the higher $N\geq 4$ we encountered another obstacle during the
determination of the boundary of the domain ${\cal D}$ of the
``admissible'' free parameters rendering the matrix of metric
positive definite. This goal appeared overambitious and hopeless. It
turned out that the construction could hardly be algebraic and/or
non-numerical. Due to the enormous growth of the unfriendliness of
secular polynomials the task of the proof of positivity of the
candidates for the metric appeared next to impossible. Fortunately,
in a climax of our paper we arrived at an innovative, feasible
resolution of the problem. The proof has been found, thanks to the
restriction of attention to the sufficiently small vicinity of the
Hermitian limit $\Theta(H) \to I$, in the omission of the
higher-powers of $\gamma$ and in the ultimate discovery of the
elementary Lie-algebraic form of the leading-order difference
$\Theta(H)-I \sim \gamma L_y+{\cal O}(\gamma^2)$ and of the
elementary factorizability of the secular polynomial.

\section*{Acknowledgement}

Work supported by GA\v{C}R Grant Nr. 16-22945S.

\end{document}